\documentclass[conference,10pt]{IEEEtran}
\IEEEoverridecommandlockouts

\usepackage{cite}
\usepackage{amsmath,amssymb,amsfonts}
\usepackage{algorithmic}
\usepackage{graphicx}
\usepackage{textcomp}
\usepackage{xcolor}
\usepackage{lipsum}
\usepackage{amsthm}
\usepackage{amsmath}
\usepackage{amssymb}
\usepackage{delarray}
\usepackage{bm}
\usepackage{graphicx}
\usepackage{color}
\usepackage{enumitem}
\usepackage{footmisc}

\usepackage{subfig}
\usepackage{graphicx}
\usepackage{cancel}
\usepackage[export]{adjustbox}
\usepackage{comment}
\usepackage{graphicx}
\usepackage{subfig}
\usepackage{caption}

\newtheorem{theorem}{Theorem}

\newtheorem{lemma}{Lemma}
\newtheorem{prop}{Proposition}

\newtheorem{rem}{Remark}

\newcommand{\cS}{\mathcal{S}}
\makeatletter
\newcommand*{\rom}[1]{\expandafter\@slowromancap\romannumeral #1@}
\makeatother

\newcommand{\rinc}[2]{f^{#1}_{#2}}

\newcommand{\inc}[2]{\bF^{#1}_{#2}}

\newcommand{\pc}[2]{\eta_{#1 #2}}
\newcommand{\cut}[1]{g_{#1}}
\newcommand{\PQ}[2]{\bP_{#1, #2}}
\newcommand{\pq}[1]{\bP_{#1}}

\renewcommand{\SS}{\mathbb{S}}
\newcommand{\PP}[1]{\mathsf{P}_{(#1)}}

\def\BibTeX{{\rm B\kern-.05em{\sc i\kern-.025em b}\kern-.08em
    T\kern-.1667em\lower.7ex\hbox{E}\kern-.125emX}}

\newcommand{\N}{n}
\renewcommand{\det}[1]{\mathrm{det}\left(#1\right)}
\newcommand{\bD}{\mathbf{D}}
\newcommand{\bP}{\mathbf{P}}
\newcommand{\bF}{\mathbf{F}}
\newcommand{\bM}{\mathbf{M}}
\newcommand{\bQ}{\mathbf{Q}}
\newcommand{\rnk}{\text{rank}}
\newcommand{\sgn}{\text{sign}}

\begin{document}

\title{When an Energy-Efficient Scheduling is Optimal for Half-Duplex Relay Networks?}
\author{
\IEEEauthorblockN{Sarthak Jain, Martina Cardone,  Soheil Mohajer}
University of Minnesota, Minneapolis, \!MN 55455, \!USA,
\!Email: \{jain0122, mcardone, soheil\}@umn.edu\\
\vspace{-0.9em}
\thanks{This research was supported in part by NSF under Award \#1907785.
}
}

\maketitle

\begin{abstract}
This paper considers a diamond network with $\N$ interconnected relays, namely a network where a source communicates with a destination by hopping information through $\N$ communicating/interconnected relays. Specifically, the main focus of the paper is
on characterizing sufficient conditions under which the $\N+1$ states (out of the $2^{\N}$ possible ones) in which at most one relay is transmitting  suffice to characterize the approximate capacity, that is the Shannon capacity up to an additive gap that only depends on $\N$. Furthermore, under these sufficient conditions, closed form expressions for the approximate capacity and scheduling (that is, the fraction of time each relay should receive and transmit) are provided. A similar result is presented for the dual case, where in each state at most one relay is in receive mode. 

\end{abstract}

\section{Introduction}

Computing the Shannon capacity of a wireless relay network is an open problem. In a half-duplex $\N$-relay network, each relay can either transmit or receive at a given time instant and therefore a \textit{scheduling} question arises: What fraction of time each relay in the network should be scheduled to receive/transmit information so that rates close to the Shannon capacity of the network can be achieved? 

For an $\N$-relay half-duplex network, there are $2^\N$ possible receive/transmit configuration states, because each relay can either be scheduled for reception or transmission. 
However, in~\cite{CardoneTIT2016}, it has been surprisingly shown that only $\N+1$ out of these $2^\N$ possible states are sufficient to achieve the network {\em approximate capacity}, i.e., an additive gap approximation of the Shannon capacity, where the gap is only a function of $\N$.
This result opens novel research directions, such as characterizing {a set} of $\N+1$ {\em critical} states for each network efficiently (in polynomial time in $\N$). 

In this work, we investigate the question above in the context of diamond networks with $\N$ interconnected relays, where the source communicates with the destination by hopping information through $\N$ half-duplex relays that can communicate with each other. 
In particular, we analyze the linear deterministic approximation of the Gaussian noise channel, and characterize sufficient {conditions
under which} at most one relay is required to transmit at any given  time {to achieve the approximate capacity.}  This leads to a significant reduction in the {average power consumption at the relays, compared to a random network with identical $n$ (where potentially at each point in time more than one relay is transmitting)
and hence,} the proposed scheduling is energy-efficient. The other advantage of a schedule with at most one relay in transmit mode is that {it simplifies} the  synchronization problem at the destination.
Our result can be readily translated to {obtain} sufficient conditions under which operating the network only in states with at most one relay in receive mode is sufficient to achieve the approximate capacity. {The proposed scheduling and approximate capacity} can be translated to obtain similar results for the practically relevant Gaussian noise channel.

To the best of our knowledge, this is the first work that provides network conditions that suffice to characterize {a set} of $\N+1$ critical states for arbitrary values of $\N$ in relay networks where, in addition to broadcasting and signal superposition, we also have communicating/interconnected relays.

\smallskip
\noindent {\bf{Related Work.}}
The cut-set bound has been shown to offer a constant (i.e., which only depends on $\N$) additive gap approximation of the Shannon capacity for Gaussian relay networks~\cite{AvestimehrIT2011,OzgurIT2013,LimIT2011,LimISIT2014,CardoneIT2014}. 
Such an approximation, for an $\N$-relay Gaussian half-duplex network can be computed by solving a linear program involving $2^\N$ cut constraints and $2^\N$ variables corresponding to the receive/transmit configurations of the $\N$ half-duplex relays.
{However, it has been shown} that it suffices to operate the network in only $\N+1$ states out of the $2^\N$ possible ones to achieve the approximate capacity~\cite{CardoneTIT2016}. Finding such a set of $\N+1$  critical states in polynomial time in $\N$ for half-duplex Gaussian relay networks is an open problem. These critical states and the approximate capacity can be computed in polynomial time {for the} following networks: (i) $\N=2$ relay half-duplex diamond networks with {\em non-interconnected} relays~\cite{bagheri2014} and {\em interconnected} relays~\cite{JainITW2021}; (ii) line networks~\cite{EzzeldinISIT2017}; (iii) a special class of layered networks~\cite{EtkindTIT2014}; and (iv) { diamond networks with $\N$ non-interconnected} relays under certain network conditions expressed in~\cite{JainISIT2019}. 
We highlight that the result presented in this paper subsumes the result for diamond networks studied in~\cite{JainISIT2019} (with no interconnection among the relays) and our recent result in~\cite{JainITW2021} for $\N=2$.

\smallskip
\noindent {\bf{Paper Organization.}} Section~\ref{sec:SystemModel} introduces the notation, describes the Gaussian and the linear deterministic half-duplex {diamond network with $n$ interconnected relays} and summarizes known capacity results. 
Section~\ref{sec:AppCapSched} presents the main result of the paper, the proof of which is in Section~\ref{sec:ProofTheorem}.
Specifically, Section~\ref{sec:AppCapSched} {characterizes sufficient} conditions under which the set of (at most) $\N+1$ network states in which at most one relay is { transmitting  (and the set with at most $n+1$ states with at most one relay receiving)} suffice to characterize the approximate capacity of the binary-valued linear deterministic approximation of the Gaussian noise channel.
{Some of the proofs can be found in the appendix.}

\section{Notation and System Model}
\label{sec:SystemModel}
\noindent \emph{Notation:} We denote the set of integers $\{i, \ldots, m\}$ by $[i:m]$, and $\{1,\ldots,m\}$ by $[m]$; note that $[i:m]=\varnothing$ if $i >m$.
For a variable $\theta$ and a set $\mathcal{X}$, $\theta_{\mathcal{X}}=\{\theta_x : x \in \mathcal{X}\}$.
We use boldface letters to refer to matrices. 
For a matrix $\bM$, $\det{\bM}$ is the determinant of $\bM$, $\bM^T$ is the matrix transpose of $\bM$ and ${\bM}_{\mathcal{A},\mathcal{B}}$ is the submatrix of  $\bM$ obtained by retaining all the rows  indexed by  the set $\mathcal{A}$ and all the columns indexed by the set $\mathcal{B}$. Matrix columns and rows are indexed beginning from $0$ (instead of $1$). $\lfloor \cdot\rfloor$  and $\lceil \cdot \rceil$ are the floor and ceiling operations, respectively, and $\left [a\right ]^+ = \max\{a,0\}$. $\mathbf{0}_{p \times q}$ is the zero matrix of dimension $p \times q$; $\mathbf{I}_p$ is the $p \times p$ identity matrix.

 The Gaussian half-duplex diamond {network with $n$ interconnected relays} consists of a source (node $s$) that wishes to communicate with a destination (node $d$) through $\N$ {\em interconnected} 
relays.
At each time instant $t$, the input/output relationship of 
this network is described as

\begin{align}
\label{eq:Diam}
\begin{split}
Y_{d}(t)&\!=\!\sum_{i=1}^{\N} S_{i}(t) h_{di} X_{i}(t) + Z_{d}(t),
\\Y_{i}(t)&\!=\!(1\!-\!S_{i}\hspace{-1pt}(t))\! \Big (\!h_{is} X_{s}\hspace{-1pt}(t)\!+\hspace{-5pt} \sum_{j \in [\N]} \!\!\! S_{j}(t) h_{ij} X_{j}\hspace{-1pt}(t)\!+\!Z_{i}\hspace{-1pt}(t) \!\Big ),  
\end{split}
\end{align}
for $i\in [\N]$. 
Note {that, at each time instant $t$:} 
(i) $S_{i}(t)$ is a binary random variable that indicates the state of relay $i \in [\N]$, with {$S_{i}(t)=0$} (respectively, {$S_i(t)=1$}) indicating that relay $i$ is receiving (respectively, transmitting); 
(ii) $X_{i}(t)$ is the channel input at node $i$ that satisfies the unit average power constraint $\mathbb{E}[|X_{i}(t)|^2] \leq 1$ for  $i \in \{s\}\cup [\N]$; 
(iii) $h_{ij}$ with $i \in [\N] \cup \{d\}$ and $j \in \{s\} \cup [\N]$ is the {\em time-invariant }
complex channel gain from node $j$ to node $i$; note that $h_{ds}=0$ and, {since the relays operate in half-duplex mode, without loss of generality we let $h_{ii}=0$;}
(iv) $Z_{i}(t)\sim \mathcal{CN}(0,1)$ is the complex additive white Gaussian noise at node $i\in\{d\} \cup [\N]$; 
and finally, (v) $Y_{i}(t)$ is the received signal at node $i\in\{d\} \cup [\N]$.

\begin{figure}[t] 
\includegraphics[width=0.7\columnwidth]{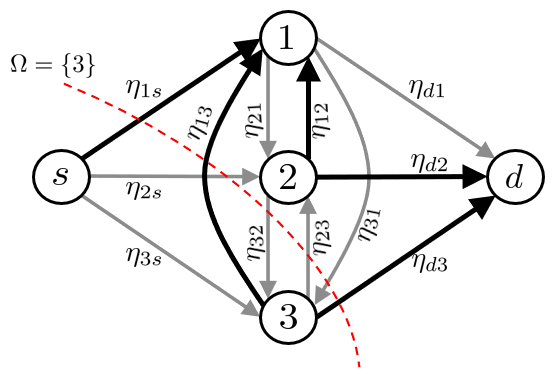}
\centering
\caption{{Diamond network with $n=3$ interconnected relays} (with cut $\Omega=\{3\}$ and state $\mathcal{S}=\{2,3\}$).}
\label{systemmodel}
\vspace{-5mm}
\end{figure}

The Shannon capacity $\mathsf{C}^{\text{G}}$ of the network  in~\eqref{eq:Diam} is not known for general $\N$. However, the capacity can be approximated within a constant {$O(n)$} bit gap.  More precisely, we can
focus on the binary {\em linear deterministic} approximation of the Gaussian noise network model~\cite{AvestimehrIT2011}, {for which the approximate capacity is known and provides}
an approximation for $\mathsf{C}^{\text{G}}$.
The linear deterministic model (a.k.a. ADT model~\cite{AvestimehrIT2011}) corresponding to the Gaussian noise network in~\eqref{eq:Diam} has an input-output relationship given by
\begin{flalign}
\label{eq:DiamLDA}
\begin{split}
& Y_{d}(t)\!=\!\sum_{i=1}^{\N} S_{i}(t) {\bD}^{\eta-\pc{d}{i}} X_{i}(t),
\\& Y_{i}(t)\!\hspace{-1pt}=\hspace{-1pt}\!(1\!-\!S_{i}(t)\hspace{-1pt}) \!\Big (\!{\bD}^{\eta-\pc{i}{s}}\hspace{-2pt} X_{s}(t)\!+\!\!\!\hspace{-2pt} \sum_{j \in [\N]} \!\!\! S_{j}(t) {\bD}^{\eta-\pc{i}{j}}\hspace{-2pt} X_{j}(t) \!\Big)\hspace{-1pt}, 
\end{split}
\end{flalign}
for $i\in [\N]$, where 
\begin{flalign*}
 {\bD}^{\eta-m} =\left[ 
\begin{array}{c|c} 
  {\mathbf{0}}_{(\eta-m)\times m} & {\mathbf{0}}_{(\eta-m)\times (\eta-m)} \\
  \hline 
  {\mathbf{I}}_{m } & {\mathbf{0}}_{m \times (\eta - m)}
\end{array} 
\right],
\end{flalign*}
and

\begin{align*}
\pc{i}{j} = \left \lceil \log |h_{ij}|^2 \right  \rceil^+, \ i \in [\N] \cup \{d\}, j \in \{s\} \cup [\N], i \neq j.
\end{align*}
Here, the vectors $X_s(t)$, $X_i(t)$, $Y_d(t)$, and $Y_i(t)$ with $i \in [\N]$ are binary of length $\eta=\max \pc{i}{j}$, where {the maximization} is taken over all {channels $\pc{i}{j}$'s} in the network; ${\bD}$ is the so-called $\eta \times \eta$ shift matrix, and $S_i(t), i \in [\N]$ is the $i$th relay binary-valued state random variable.

The approximate capacity of the linear deterministic model in~\eqref{eq:DiamLDA}  is given by the solution of 

\begin{align} \label{maxmincut}
\begin{split}
    \mathsf{C}^{\text{LD}}= \max_{\bm{\lambda}}  \ & t  \\
    \text{s.t.\  }     & t \leq 
    \cut{\Omega} \triangleq \sum_{\mathcal{S} \subseteq [\N]} \lambda_{\mathcal{S}} \rinc{\Omega}{\mathcal{S}},  \qquad \forall \Omega \subseteq [\N], \\
      &\cut{p} \triangleq \sum_{\mathcal{S}\subseteq{[\N]}} \lambda_{\mathcal{S}} \leq 1,\\ &\lambda_{\mathcal{S}} \geq 0, \qquad \forall  \mathcal{S} \subseteq[\N],
\end{split}
\end{align}
{where:} (i) $\mathcal{S} =\{i\in [\N] : S_i=1\}$ is the set of relay nodes in transmit mode; 
(ii) $\lambda_{\mathcal{S}} \geq 0$ is the fraction of time that the network operates in state $\mathcal{S}$ and hence, $\sum_{\mathcal{S} \subseteq [\N]} \lambda_{\mathcal{S}} {\leq} 1$;
(iii) $\bm{\lambda}$ is referred to as a network {\em schedule} and is a vector obtained by stacking together $\lambda_{\mathcal{S}}$ for all  $\mathcal{S} \subseteq [\N]$;
(iv) $\Omega \subseteq [\N]$ denotes a partition of the relays in the `side of $s$', i.e., $\{s\} \cup \Omega$ is a network cut; 
similarly, $\Omega^c = [\N] \setminus \Omega$  is a partition of the relays in the `side of~$d$'. Moreover, we define 
\begin{align}\label{eq:def:pc}
    \rinc{\Omega}{\mathcal{S}} \!\triangleq\!  I \left (X_s, X_{\Omega \cap \mathcal{S}}; Y_d, Y_{\Omega^c \cap \mathcal{S}^c}|X_{\Omega^c \cap \mathcal{S}},\mathcal{S} \right ) \!=\! \text{rank} \left (\inc{\Omega}{\mathcal{S}} \right ), 
\end{align}
where  
$\inc{\Omega}{\mathcal{S}}$ {is the transfer matrix} from $X_{\{s\} \cup (\Omega \cap \cS)}$ to $Y_{\{d\} \cup (\Omega^c \cap \cS^c)}$, corresponding to the ADT model~\cite{AvestimehrIT2011}.

It turns out that  $|\mathsf{C}^{\text{G}} - \mathsf{C}^{\text{LD}}| \leq \kappa$, where $\kappa = O(n)$ is independent of the channel gains and operating SNR and hence, $\mathsf{C}^{\text{LD}}$ {in~\eqref{maxmincut}} provides an approximation for the Shannon capacity of the network in~\eqref{eq:Diam}\footnote{We highlight that schemes such as quantize-map-and-forward~\cite{AvestimehrIT2011} and noisy network coding~\cite{LimIT2011}, together with the cut set bound, allow to characterize the capacity of Gaussian relay networks up to a constant additive gap.}.

\smallskip
\noindent{\bf{Example 1.}} Consider the diamond network with $\N=3$ {interconnected} relays in Fig.~\ref{systemmodel}. For {the cut} $\Omega=\{3\}$ and state $\mathcal{S}=\{2,3\}$, we have $\{s\}\cup (\Omega \cap \cS) = \{s,3\}$ and $\{d\}\cup (\Omega^c \cap \cS^c) = \{d,1\}$. The input-output relationship for this cut and state is given by

\begin{align*}
\begin{bmatrix} 
Y_d\\
Y_1
\end{bmatrix} = \begin{bmatrix}
{\bD}^{\eta} & {\bD}^{\eta - \pc{d}{3}}
\\
{\bD}^{\eta - \pc{1}{s}} & {\bD}^{\eta - \pc{1}{3}}
\end{bmatrix}
\begin{bmatrix} 
X_s\\
X_3
\end{bmatrix}.
\end{align*}
Therefore, we have 
\begin{align*}
\qquad \rinc{\{3\}}{\{2,3\}} = \mathrm{rank}\left(\inc{\{3\}}{\{2,3\}}\right) = \mathrm{rank}
\begin{bmatrix}
{\bD}^{\eta} & {\bD}^{\eta - \pc{d}{3}}
\\
{\bD}^{\eta - \pc{1}{s}} & {\bD}^{\eta - \pc{1}{3}}
\end{bmatrix}.\quad \diamond 
\end{align*}
In this work, we seek to identify {\em sufficient} network conditions which allow {to 
determine} {a set} of {$\N+1$ states} (out of the $2^\N$ possible ones) that suffice to achieve the {approximate} capacity in~\eqref{maxmincut} of the linear deterministic network {and can} be readily translated {into} a similar result for  the original noisy Gaussian channel model in~\eqref{eq:Diam}.

\section{Main Result: Conditions for Optimality of States with at Most One Relay Transmitting}
\label{sec:AppCapSched}

Without loss of generality, we assume {that} the relay nodes are arranged in increasing order of their left link capacities, that is, $\pc{1}{s} \leq \pc{2}{s} \leq \dots \leq \pc{\N}{s}$. We define $\pq{}$ to be an $(\N+2) \times (\N+2)$ matrix, {the rows and columns of which are indexed by $[0:n+1]$,} 
and 
\begin{equation}
\label{eq:MatrP}
    \PQ{i}{j}=  
\begin{cases}
-\rinc{[i:\N]}{\{j\}}, \quad &(i,j) \in [\N+1]^2,\\
0, 
\quad &(i,j)=(0,0),\\
1, \quad &\text{otherwise,}
\end{cases}
\end{equation}
where we define  $\rinc{\Omega}{\{\N+1\}}=\rinc{\Omega}{\varnothing}$, for consistency. Moreover, for $i\in [0:\N+1]$ we use $\PP{i}$ to denote the \emph{minor} of $\bP$ associated with the row $0$ and column $i$ {of the} matrix $\bP$, that is 
\[
\PP{i} \triangleq \det{\PQ{{[\N+1]}}{[0:\N+1]\setminus \{i\}} }. 
\]
Finally, we define 
\[
 \SS\triangleq \{\{1\}, \{2\}, \dots, \{\N\}, \varnothing \},
\]
to be the set of the $\N+1$ states, where at most one relay is transmitting in each state.

The main result of this paper is presented in Theorem~\ref{thm1}, which characterizes sufficient network conditions for the optimality of operating the network only in states $\cS\in \SS$.
\begin{theorem} \label{thm1}
{Whenever} ${\det{\bP} \neq 0}$ and ${\frac{(-1)^{\N+1} \PP{\N+1}}{ \det{\bP}} \geq 0}$, then it is optimal to operate the network in 
states in $\SS$ to achieve $\mathsf{C}^{\text{LD}}$ in~\eqref{maxmincut}.
\end{theorem}

\noindent
{\bf{Example 2.}}
Consider the relay network in Fig.~\ref{systemmodel} with link capacities given by ${\pc{1}{s}=1},$ ${\pc{2}{s}=3},$ ${\pc{3}{s}=5},$ ${\pc{d}{1}=6},$ ${\pc{d}{2}=5},$ ${\pc{d}{3}=3},$ ${\pc{1}{2}=3},$ ${\pc{2}{1}=4},$ ${\pc{3}{2}=5},$ ${\pc{2}{3}=3},$ ${\pc{3}{1}=2}$ and ${\pc{1}{3}=4}$.
For this network, the matrix $\bP$ is given~by
\begin{equation}\label{example}
    \bP=
\begin{bmatrix}
0 & 1 & 1 & 1 & 1\\
1 & -6 & -5 & -3 & 0\\
1 & 0 & -6 & -4 & -1 \\
1 & -3 & -1 & -7 & -3 \\
1 & -5 & -5 & -3 & -5
\end{bmatrix}.
\end{equation}
From this one can  verify that ${\det{\bP} =280 \neq 0}$ and ${ \frac{(-1)^4 \PP{4}}{\det{\bP}} = \frac{\det{\PQ{[4]}{[0:3]}}}{\det{\bP}} =\frac{8}{280} \geq 0}$, i.e., the conditions in Theorem~\ref{thm1} are satisfied for this $n=3$ relay network.
Thus, operating this network in ${\SS=\{\{1\}, \{2\}, \{3\}, \varnothing\}}$ achieves $\mathsf{C}^{\text{LD}}$ in~\eqref{maxmincut}.\hfill $\diamond$

\begin{rem}
Note that $\SS$ consists of all the states where at most one relay is transmitting, while the rest of the relays are receiving. A similar  condition can be obtained for the optimality of the states 
\[
\SS'=\{[\N], [\N]\setminus \{1\}, [\N]\setminus \{2\}, \dots, [\N]\setminus \{\N\}\},
\] 
where at most one relay is in receive mode. 

\end{rem}

\begin{rem}
For the case when the relays are non-interconnected, i.e., {$\pc{i}{j}=0$ for all $(i,j) \in [\N]^2$,} the result in Theorem~\ref{thm1} subsumes the result in~\cite{JainISIT2019}. Moreover, for $\N=2$, the result in Theorem~\ref{thm1} is equivalent to the one in~\cite{JainITW2021}, where we characterized the set of at most $3$ states that suffice to achieve $\mathsf{C}^{\text{LD}}$ in~\eqref{maxmincut} for $\N=2$.
\end{rem}

\begin{rem}
{The conditions in Theorem~\ref{thm1} are a consequence of the relaying scheme used to operate the network in states $\SS$. This scheme is based on information flow preservation at each relay, i.e., the amount of unique linearly independent bits that each relay decodes is equal to the amount of unique linearly independent bits that each relay transmits. The conditions in Theorem~\ref{thm1} ensure the feasibility of this scheme.}
\end{rem}

In the remainder of this section, we analyze {the} variables $\rinc{\Omega}{\cS}$ and present some of {their} properties, which play an important role in the proof of Theorem~\ref{thm1}, presented in Section~\ref{sec:ProofTheorem}.

\subsection{Properties of $\rinc{\Omega}{\mathcal{S}}$}
We here present two properties of $\rinc{\Omega}{\mathcal{S}}=\text{rank}(\inc{\Omega}{\mathcal{S}})$ that we will leverage in the proof of Theorem~\ref{thm1}.
\begin{prop}
\label{prop:Propfs}
For all $\Omega \subseteq [\N]$ and $\mathcal{S}\subseteq [\N]$, we have that
    \begin{equation} \label{ranks}
        \rinc{\Omega}{\mathcal{S}} \geq \max_{i \in \Omega^c \cap \mathcal{S}^c} \pc{i}{s} +\max_{j \in \Omega \cap \mathcal{S}} \pc{d}{j},
    \end{equation}
with equality if $\Omega^c \cap \mathcal{S}^c = \varnothing$ or $\Omega \cap \mathcal{S} = \varnothing$.
\end{prop}

\begin{proof}
Let
\[i^\star = \arg\max_{i\in \Omega^c \cap \cS^c} \pc{i}{s},\quad \textrm{and}  \quad  j^\star = \arg\max_{j\in \Omega \cap \cS} \pc{d}{j}.
\]
Then, the submatrix of $\inc{\Omega}{\mathcal{S}}$ induced by row blocks $\{d,i^\star\}$ and column blocks $\{s,j^\star\}$ is 
\begin{align*}
    \left[ 
\begin{array}{c|c} 
  \mathbf{D}^{\eta-\pc{d}{s}} & \mathbf{D}^{\eta-\pc{d}{j^\star}} \\ 
  \hline 
  \mathbf{D}^{\eta-\pc{i^\star}{s}} & \mathbf{D}^{\eta-\pc{i^\star}{j^\star}}
\end{array} 
\right] 
=
  \left[ 
\begin{array}{c|c} 
  {\mathbf{0}_{\eta \times \eta}} & \mathbf{D}^{\eta-\pc{d}{j^\star}} \\ 
  \hline 
  \mathbf{D}^{\eta-\pc{i^\star}{s}} & \mathbf{D}^{\eta-\pc{i^\star}{j^\star}}
\end{array} 
\right], 
\end{align*}
the {rank of which} is at least $\pc{i^\star}{s} + \pc{d}{j^\star}$. This provides a lower bound {on the} rank of $\inc{\Omega}{\mathcal{S}}$. 
 Moreover, if $\Omega^c \cap \mathcal{S}^c = \varnothing$, then $\inc{\Omega}{\mathcal{S}}$ only consists of one row induced by $\{d\}$ and columns induced by $\{\pc{d}{j}: j \in \Omega \cap \mathcal{S}\}$  and hence, $\text{rank}(\inc{\Omega}{\mathcal{S}})=\max_{j \in \Omega \cap \mathcal{S}} \pc{d}{j}$. Similarly, if $\Omega \cap \mathcal{S} = \varnothing$, then $\inc{\Omega}{\mathcal{S}}$ has only one column and hence,  $\text{rank}(\inc{\Omega}{\mathcal{S}})=\max_{i \in \Omega^c \cap \mathcal{S}^c} \pc{i}{s}$.
{This concludes the proof of Proposition~\ref{prop:Propfs}.}
 \end{proof}

\begin{prop} \label{submodularity}
For a given state $\mathcal{S} \subseteq [\N]$, $\rinc{\Omega}{\mathcal{S}}$ is submodular in $\Omega$, that is,
\begin{equation}
    \rinc{\Omega_1}{\mathcal{S}}+\rinc{\Omega_2}{\mathcal{S}} \geq \rinc{\Omega_1 \cap \Omega_2}{\mathcal{S}}+\rinc{\Omega_1 \cup \Omega_2}{\mathcal{S}},
\end{equation}
for any subsets $\Omega_1,\Omega_2 \subseteq [\N]$.
Similarly, for a given cut ${\Omega \subseteq [\N]}$, $\rinc{\Omega}{\mathcal{S}}$ is submodular in $\mathcal{S}$, that is, \[\rinc{\Omega}{\mathcal{S}_1}+\rinc{\Omega}{\mathcal{S}_2} \geq \rinc{\Omega}{\mathcal{S}_1 \cup \mathcal{S}_2}+\rinc{\Omega}{\mathcal{S}_1 \cap \mathcal{S}_2},\] {for any subsets $\mathcal{S}_1,\mathcal{S}_2 \subseteq [\N]$.}
\end{prop}
\begin{proof}
{The proof is given  in~\cite{CardoneTIT2016}.}
\end{proof}

\section{Proof of Theorem~\ref{thm1}}
\label{sec:ProofTheorem}
In this section, we present the proof of Theorem~\ref{thm1}, which consists of three main steps. 
We first introduce an auxiliary optimization problem {in~\eqref{maxmincutupper}} in Section~\ref{sec:upperbound}, {which is a} relaxed version of the problem in~\eqref{maxmincut} and hence, its solution $\mathsf{C}^{\text{U}}$ provides an upper bound {on the} optimum value of~\eqref{maxmincut}, i.e.,  $\mathsf{C}^{\text{U}} \geq \mathsf{C}^{\text{LD}}$.  In Section~\ref{upperbound}, we propose a solution $(\bm{\lambda}^\star,t^\star)$ for the (relaxed) optimization problem in~\eqref{maxmincutupper},  where ${\lambda}^\star_\cS =0$ for all $\cS\notin \SS$. We show  that, under the {conditions} of Theorem~\ref{thm1}, $(\bm{\lambda}^\star, t^\star)$ is  feasible and optimal,  which leads to   $\mathsf{C}^{\text{U}}= t^\star$.
{Finally,} in Section~\ref{feasibilityofapprox} we  show that the proposed solution is feasible for the original problem in~\eqref{maxmincut}, implying that $\mathsf{C}^{\text{LD}} \geq t^\star$. Therefore, putting the three results together, we get  $ t^\star = \mathsf{C}^{\text{U}} \geq  \mathsf{C}^{\text{LD}} \geq t^\star$. This shows that $\mathsf{C}^{\text{LD}}= t^\star$ {can} be attained by $\bm{\lambda}^\star$ that satisfies the claim of Theorem~\ref{thm1}. This concludes the proof of the theorem. 

\subsection{An Upper Bound for the Approximate Capacity}
\label{sec:upperbound}
The optimization problem in~\eqref{maxmincut} consists of $2^\N$ cut constraints. We can relax these constraints except for $\N+1$ of them.  More formally, we define 
\begin{align} \label{maxmincutupper}
\begin{split}
    \mathsf{C}^{\text{U}}= \max_{\bm{\lambda}}  \ & t  \\
    \text{s.t. } & t  \leq \cut{[i:\N]} =  \sum_{\mathcal{S} \subseteq [\N]} \lambda_{\mathcal{S}} \rinc{[i:\N]}{\mathcal{S}}, \quad i \in [\N+1],  \\
     & \cut{p} \triangleq \sum_{\mathcal{S}{\subseteq{[\N]}}} \lambda_{\mathcal{S}} \leq 1,\\
     &\lambda_{\mathcal{S}} \geq 0,  \qquad   \mathcal{S} \subseteq[\N]  ,
\end{split}
\end{align}
where ${[\N+1:\N]}={\varnothing}$.

Note that the optimization problem in~\eqref{maxmincutupper} is less constrained compared to~\eqref{maxmincut}. Hence, {the problem in~\eqref{maxmincutupper}} has a wider feasible set {than~\eqref{maxmincut},} and its maximum objective function cannot be smaller than that of the {problem in~\eqref{maxmincut}}, i.e., $\mathsf{C}^{\text{U}} \geq \mathsf{C}^{\text{LD}}$.

\subsection{An Optimal Solution for the {Optimization Problem in~\eqref{maxmincutupper}}} \label{upperbound}
We show that under {the conditions of Theorem~\ref{thm1},} the states {in $\SS$} are optimal for the upper bound $\mathsf{C}^{\text{U}}$ on the approximate capacity obtained by solving~\eqref{maxmincutupper}, that is, there exists an optimal solution with $(t,\lambda_{\SS}) \geq 0$ and $\lambda_{\mathcal{S}} = 0$, for all  $\mathcal{S} \notin \SS$. 
In particular, {we start by stating the following proposition, the proof of which can be found in Appendix~\ref{NewproofThm1}.}
\begin{prop}
\label{prop:OptUB:1}
Assume $\det{\bP} \neq 0$ and $\frac{(-1)^{\N+1} \PP{\N+1}}{ \det{\bP}} \geq 0$. Then, the variables 
\begin{align}
\label{eq:OptUB}
\begin{split}
        \lambda^\star_{\varnothing} &:= \lambda^\star_{\{n+1\}} = (-1)^{n+1}\frac{\PP{n+1} }{\det{\bP}},\\
        \lambda_{\{i\}}^\star &=
        (-1)^{i}\frac{\PP{i} }{\det{\bP}}, \qquad i\in [\N], \\
        \lambda_{\cS}^{\star} &=0, \qquad \cS \notin \SS,
        \end{split}
\end{align}
are non-negative. 
\end{prop}
{We now leverage the {Karush–Kuhn–Tucker (KKT)} conditions to prove the following proposition.}
\begin{prop}
\label{prop:OptUB}
Assume $\det{\bP} \neq 0$ and $\frac{(-1)^{\N+1} \PP{\N+1}}{ \det{\bP}} \geq 0$. Then,  $\bm{\lambda}^\star=\{\lambda_{\cS}:\cS\subseteq [n]\}$ defined in~\eqref{eq:OptUB}  is an optimal solution for {the 
problem} in~\eqref{maxmincutupper}. 
Consequently, we have 
\begin{align}\label{eq:optUB-t}
\mathsf{C}^{\text{U}} = t^\star=  \cut{[i:\N]}^\star = \sum_{\cS \in \SS} \lambda_{\cS}^\star \rinc{[i:\N]}{\mathcal{S}},
\end{align}
for every $i\in [\N+1]$. 
\end{prop}

\begin{proof}[Proof of  Proposition~\ref{prop:OptUB}]
The {proof leverages} the KKT conditions. For the KKT multipliers ${\bm{\mu}=(\mu_p, \mu_1,\dots,\mu_{\N+1})}$ and ${\bm{\sigma}=(\sigma_{\mathcal{S}}:  \mathcal{S} \subseteq [\N])}$, the Lagrangian for the optimization problem in~\eqref{maxmincutupper} is given~by
\begin{align} \label{lagrangian}
\begin{split}
  \mathcal{L}(\bm{\mu}, \bm{\sigma}, \bm{\lambda}, t)=&-t +  \sum_{i\in [\N+1]} \mu_i (t-g_{[i:\N]}) 
\\& + \mu_p (g_p-1) 
  - \sum_{\mathcal{S}\subseteq [\N]}  \sigma_\mathcal{S} \lambda_{\mathcal{S}} .
\end{split}
\end{align}
In the following, we proceed with a choice of $(\bm{\mu}, \bm{\sigma})$ where $\bm{\mu}$ is the solution of 
\begin{align} \label{eq:mu}
     \begin{bmatrix}
\mu_p & \mu_{1}  & \ldots & \mu_{\N} & \mu_{\N+1} \end{bmatrix} \pq{} \!=\! 
\begin{bmatrix} 1  & 0 & \ldots & 0 & 0 \end{bmatrix} \!\!,
\end{align}
and 
\begin{align}\label{eq:sigma}
    \sigma_{\cS} =
        \mu_p - \sum_{i=1}^{\N+1} \mu_i \rinc{[i:\N]}{\cS},
\end{align}
for every $\cS\subseteq [\N]$. 
{We next prove} the optimality of $(\bm{\lambda}^\star, t^\star)$ by showing that the set of KKT multipliers $(\bm{\mu}, \bm{\sigma})$ defined in~\eqref{eq:mu} and~\eqref{eq:sigma} together with  $(\bm{\lambda}^\star, t^\star)$ defined in~\eqref{eq:OptUB} and~\eqref{eq:optUB-t} satisfy the following four groups of conditions. 

\smallskip
\noindent $\bullet$ {\em{Primal Feasibility.}} First, note that {Proposition~\ref{prop:OptUB:1}} guarantees that the  constraint  ${\lambda_\cS^\star \geq 0}$ is satisfied for every ${\cS\subseteq [n]}$.  In order to show the feasibility of the solution, it remains {to show 
that} $t^\star \leq \cut{[i:n]}$ for every $i\in[n+1]$ and $\sum_{\cS\subseteq [n]} \lambda^\star_\cS \leq 1$. Note that by forcing {these inequalities} to hold with equality, and setting  $\lambda_{\mathcal{S}}  =0$, for all $\mathcal{S} \subseteq [\N]$ with  $\mathcal{S} \neq \SS$ in~\eqref{maxmincutupper}, we obtain a system of $(\N+2)$ linear equations in $(\N+2)$ variables (namely $t$ and $\lambda_\cS$ for $\cS\in\SS$), given by
\begin{align*} 
    &\pq{} \begin{bmatrix}
t & \lambda_{\{1\}}  & \ldots & \lambda_{\{\N\}} & {\lambda_{\varnothing}} \end{bmatrix}^T \!=\! 
\begin{bmatrix} 1  & 0 & \ldots & 0 & 0 \end{bmatrix}^T \!\!,
\end{align*}
where $\pq{}$ is the matrix defined in~\eqref{eq:MatrP}. The solution of this system of linear {equations} is indeed given in~\eqref{eq:OptUB} and~\eqref{eq:optUB-t}. 
Therefore, the solution ${(\bm{\lambda}^\star, t^\star)}$ is feasible for the optimization {problem} in~\eqref{maxmincutupper}.  In particular, it is worth noting that $t^\star = \sum_{\cS \in \SS} \lambda_{\cS}^\star \rinc{[i:\N]}{\mathcal{S}}$ is guaranteed by the  $i$th row of the matrix identity above, and hence~\eqref{eq:optUB-t} holds for  all values of $i\in[n+1]$.

\smallskip
\noindent $\bullet$ {\em{Complementary Slackness.}} 
{We need} to show that $(\bm{\mu}, \bm{\sigma})$ and the solution $(\bm{\lambda}^\star, t^\star)$ given in~\eqref{eq:OptUB} satisfy 
\begin{itemize}
    \item $\mu_i (t^\star-\cut{[i:\N]}^\star) =0$  for all $i \in [\N+1]$;
    \item $\mu_p (\cut{p}^\star-1)=0$;
    \item and $\sigma_\cS \lambda_\cS^\star = 0$ for every $\cS\subseteq [\N]$. 
\end{itemize}
The first and the second conditions are readily implied by~\eqref{eq:OptUB} for any choice of $\bm{\mu}$. Moreover, the third condition holds for $\cS\notin \SS$, since we have $\lambda_{\cS}^\star=0$ whenever $\cS\notin \SS$. Finally, consider some $\cS\in \SS$, say $\cS=\{j\}$ where $j\in [\N+1]$ (and $j=\N+1$ if $\cS=\varnothing$). Then, the definition of {$\sigma_{\{j\}}$} in~\eqref{eq:sigma} and the $j$th column of the matrix identity in~\eqref{eq:mu} imply that 
\[\sigma_{\{j\}} = \mu_p - \sum_{i=1}^{\N+1} \mu_i \rinc{[i:\N]}{\{j\}}= \bm{\mu}\cdot \bP_{[0:\N+1],j}= 0.
\]  
This ensures that $\sigma_{\mathcal{S}} \lambda_{\mathcal{S}}^\star=0$, for all $\mathcal{S} \in \SS$. 

\smallskip
\noindent $\bullet$ {\em{Stationarity.}} 
{ We aim to prove that, when evaluated in $\bm{\mu}$ as in~\eqref{eq:mu} and $\sigma_{\mathcal{S}}$ in~\eqref{eq:sigma},} the derivatives of the Lagrangian in~\eqref{lagrangian} with respect to $t$ and {$\lambda_{\mathcal{S}}, \mathcal{S} \subseteq [n]$, are zero.} By taking {the} derivative of $\mathcal{L}(\bm{\mu}, \bm{\sigma}, \bm{\lambda}^\star, t^\star)$ with respect to $t$ we get 
\begin{equation*} 
    \begin{split}
        \frac{\partial }{\partial t} \mathcal{L}(\bm{\mu}, \bm{\sigma}, \bm{\lambda}^\star, t^\star) &= -1+\sum_{i=1}^{\N+1} \mu_i \\
        &= -1 + \bm{\mu}\cdot \bP_{[0:\N+1],0} = -1+1=0. 
    \end{split}
\end{equation*}
Similarly, by taking {the} derivative with respect to $\lambda_\cS$ we get
\begin{align*}
    \frac{\partial }{\partial \lambda_{\cS}}\mathcal{L}(\bm{\mu}, \bm{\sigma}, \bm{\lambda}^\star, t^\star)&=\mu_p-\sum_{i=1}^{\N+1} \mu_i \rinc{[i:\N]}{\cS} -\sigma_{\cS}\\
    &= \sigma_\cS - \sigma_{\cS} =0,
\end{align*}
in which we used the definition of $\sigma_{\cS}$ in \eqref{eq:sigma}. 

\smallskip
\noindent $\bullet$ {\em{Dual Feasibility.}} 
In {this} last part, we need to prove that the KKT {multipliers in~\eqref{eq:mu} and~\eqref{eq:sigma}} are non-negative. First, we present the following proposition, { the proof of which can be found in Appendix~\ref{app:nonnegativelvar}.}

\begin{prop}\label{prop:mu}
All the entries of the vector $\bm{\mu}$ obtained from~\eqref{eq:mu} are non-negative. 
\end{prop}


Next, {we focus on the KKT multipliers $\sigma_{\mathcal{S}}, \mathcal{S} \subseteq [n]$ in~\eqref{eq:sigma}. For} an arbitrary state $\cS=\{a_1,a_2,\dots, a_k\}$, we can write 
\begin{align*}
\sum_{i=1}^{\N+1}  &\mu_i  \rinc{[i:\N]}{\cS} = 
   \sum_{i=1}^{\N+1}  \mu_i  \rinc{[i:\N]}{\{a_1, \ldots, a_k\}}\\
   &= \sum_{i=1}^{\N+1} \mu_i \left[\sum_{j=1}^{k-1} \left(\rinc{[i:\N]}{\{a_1, \ldots, a_{j+1}\}} -   \rinc{[i:\N]}{\{a_1, \ldots, a_{j}\}}\right) + \rinc{[i:\N]}{\{a_1\}} \right]\\
   &\leq  \sum_{i=1}^{\N+1} \mu_i \left[\sum_{j=1}^{k-1} \left(\rinc{[i:\N]}{\{ a_{j+1}\}} -   \rinc{[i:\N]}{\varnothing}\right) + \rinc{[i:\N]}{\{a_1\}} \right]\\
   &=  \sum_{i=1}^{\N+1} \mu_i \left[\left(\sum_{j=1}^{k} \rinc{[i:\N]}{\{ a_{j}\}}\right) -  (k-1) \rinc{[i:\N]}{\varnothing} \right]\\
   &= \sum_{j=1}^k \sum_{i=1}^{\N+1} \mu_i  \rinc{[i:\N]}{\{ a_{j}\}} -  (k-1) \sum_{i=1}^{\N+1} \mu_i \rinc{[i:\N]}{\varnothing}\\
   &= \sum_{j=1}^k \mu_p - (k-1)\mu_p = \mu_p,
\end{align*}
where the inequality {is due to Proposition~\ref{submodularity}, i.e., $\rinc{\Omega}{\cS}$ is submodular in $\cS$. Thus,} we get $\sigma_\cS = \mu_p - \sum_{i=1}^{\N+1}  \mu_i  \rinc{[i:\N]}{\cS} \geq 0$. This concludes the proof of Proposition~\ref{prop:OptUB}.
\end{proof}

\subsection{Feasibility of $(\bm{\lambda}^\star, t^\star)$ for $\mathsf{C}^{\text{LD}}$} 
\label{feasibilityofapprox}
In Section~\ref{upperbound}, we have shown that the solution $(\bm{\lambda}^\star, t^\star)$ given in~\eqref{eq:OptUB} is optimal for the optimization problem in~\eqref{maxmincutupper}.
This implies that $t^\star = \mathsf{C}^{\text{U}} \geq \mathsf{C}^{\text{LD}}$ where $\mathsf{C}^{\text{LD}}$ is the approximate capacity of the network, obtained by solving the problem in~\eqref{maxmincut}.
In the following, we aim to prove {that} $(\bm{\lambda}^\star, t^\star)$ {in~\eqref{eq:OptUB}} is a feasible solution for the optimization problem in~\eqref{maxmincut}, which in turn implies $\mathsf{C}^{\text{LD}}\geq t^\star$, and concludes the proof of Theorem~\ref{thm1}.  Towards this end, it  suffices to  show the feasibility of {such a} solution for~\eqref{maxmincut}, as stated in the following proposition.
\begin{prop}
\label{prop:FeasiAC}
The solution $(\bm{\lambda}^\star, t^\star)$ given in~\eqref{eq:OptUB} is feasible for~\eqref{maxmincut} and thus, $\mathsf{C}^{\text{LD}}\geq  t^\star$.
\end{prop}
\begin{proof}
Note that the two optimization {problems} in~\eqref{maxmincut} and~\eqref{maxmincutupper} have identical objectives and similar constraints. More precisely, the constraints in~\eqref{maxmincutupper} are a subset of those in~\eqref{maxmincut} and hence, they are clearly satisfied for \eqref{maxmincut} {also because} $(\bm{\lambda}^\star, t^\star)$ is an optimum solution for~\eqref{maxmincutupper}. Thus, we {only need to focus} on the constraints of the form $t^\star \leq g^\star_\Omega = \sum_{\cS \subseteq [\N]} \lambda^\star_{\cS} \rinc{\Omega}{\cS}$, which exclusively appear in~\eqref{maxmincut}. 

Towards this end, {we consider} an arbitrary cut ${\Omega=[a_1:b_1] \cup [a_2:b_2] \cup \ldots \cup [a_k:b_k] \subseteq [\N]}$, 
where ${ a_1 \leq b_1 \leq a_2 \leq b_2 \leq \ldots \leq a_k \leq b_k}$.  We also define $a_{k+1}=\N+1$. 
Recall from  Proposition~\ref{submodularity} that, for a given state $\mathcal{S}$, the function $\rinc{\Omega}{\mathcal{S}}$ is submodular in $\Omega$.
Then, for $\Omega_1 = \Omega\cup [a_{i+1}:\N]$ and $\Omega_2 = [b_i+1:\N]$ {with $i \in [k]$,} we have $\Omega_1 \cup \Omega_2 = \Omega\cup [a_{i}:\N]$ and $\Omega_1 \cap \Omega_2 = [a_{i+1}:\N]$. Thus, 
\begin{align}\label{eq:subm-interval}
    \rinc{\Omega\cup [a_{i+1}:\N]}{\cS} + \rinc{[b_i+1:\N]}{\cS} \geq 
    \rinc{\Omega\cup [a_{i}:\N]}{\cS} + \rinc{[a_{i+1}:\N]}{\cS},
\end{align}
for every $i\in [k]$. Moreover,  since $a_{k+1}=\N+1$,  we have  ${[a_{k+1}:\N]=\varnothing}$,  and ${\Omega \subseteq [a_1:\N]}$.  Thus, we obtain
\begin{align*}
&\sum_{\cS\subseteq [\N]} \lambda^\star_\cS \rinc{\Omega}{\mathcal{S}} \\
&=\sum_{\cS\subseteq [\N]} \lambda^\star_\cS \left[ \sum_{i=1}^k
\left( \rinc{\Omega\cup [a_{i+1}:\N]}{\cS} - \rinc{\Omega\cup [a_{i}:\N]}{\cS}\right) + \rinc{\Omega\cup [a_{1}:\N]}{\cS}\right]\\
&\geq \sum_{\cS\subseteq [\N]} \lambda^\star_\cS \left[ \sum_{i=1}^k
\left( \rinc{[a_{i+1}:\N]}{\cS} - \rinc{[b_{i}+1:\N]}{\cS}\right) + \rinc{[a_{1}:\N]}{\cS}\right]\\
&= \sum_{i=1}^{k+1} \sum_{\cS\subseteq [\N]} \lambda^\star_\cS 
 \rinc{[a_{i}:\N]}{\cS} -  \sum_{i=1}^{k} \sum_{\cS\subseteq [\N]} \lambda^\star_\cS \rinc{[b_{i}+1:\N]}{\cS}\\
 &= \sum_{i=1}^{k+1} {g^\star_{[a_i:\N]}} -  \sum_{i=1}^{k} {g^\star_{[b_i+1:\N]}}\\
 &= (k+1)t^\star - kt^\star = t^\star, 
\end{align*}
where the inequality {follows from~\eqref{eq:subm-interval}.} This implies that ${t^\star \leq \sum_{\cS\subseteq [\N]} \lambda^\star_\cS \rinc{\Omega}{\mathcal{S}}}$ {for any} $\Omega \subseteq [\N]$.
{This concludes the claim of Proposition~\ref{prop:FeasiAC} and completes the proof of Theorem~\ref{thm1}.} 
\end{proof}

\appendices

\section{Proof of  Proposition~\ref{prop:OptUB:1}}\label{NewproofThm1}

{We start by noting that, under the conditions of Theorem~\ref{thm1},
the value of}  $\lambda_{\varnothing}^\star$ in Proposition~\ref{prop:OptUB} is readily non-negative,  i.e.,
\begin{align*}
\lambda_{\varnothing}^\star=\lambda_{\{\N+1\}}^\star=
        (-1)^{\N+1}\frac{\PP{\N+1} }{\det{\bP}} \geq 0.
\end{align*}
Thus, in what follows we focus on showing that $\lambda_{\{i\}}^\star \geq 0$ for all $i \in [\N]$.
Towards this end, we first highlight that $(\bm{\lambda}^\star, t^\star)$ in Proposition~\ref{prop:OptUB} {is the solution of a system} of linear equations constructed as follows:
(i) setting  $\lambda_{\mathcal{S}}  =0$, for all $\mathcal{S} \subseteq [\N]$ with  $\mathcal{S} \neq \SS$ in~\eqref{maxmincutupper}; and (ii) forcing constraints $t\leq  \cut{[i:\N]}$ for $i \in [\N+1]$ {and $g_p\leq 1$} in~\eqref{maxmincutupper} to hold with equality. This system of $(\N+2)$ linear equations in $(\N+2)$ variables, is given by 
\begin{align} 
\label{schemeflowpre}
    &\pq{}\! \begin{bmatrix}
t & \lambda_{\{1\}}  & \ldots & \lambda_{\{\N\}} & {\lambda_{\varnothing}} \end{bmatrix}^T \!\!=\!\! 
\begin{bmatrix} 1  & 0 & \ldots & 0 & 0 \end{bmatrix}^T \!\!\!,
\end{align}
and hence, the equation corresponding to the row $i+1$ 
for $i \in [0:n]$ of~\eqref{schemeflowpre} is given by
\begin{align*} 
    t^\star &=\sum_{j=1}^{\N+1} \lambda_{\{j\}}^\star \rinc{[i+1:\N]}{\{j\}}\\
        &\stackrel{{\rm{(a)}}}{=}\lambda_{\varnothing}^\star \pc{i}{s} + \sum_{j=1}^{i-1} \lambda_{\{j\}}^\star \pc{i}{s} + \lambda_{\{i\}}^\star \pc{(i-1)}{s}+\sum_{j=i+1}^\N \lambda_{\{j\}}^\star \rinc{[i+1:\N]}{\{j\}} \\
    &\stackrel{{\rm{(b)}}}{=} \left(1-\sum_{j=i}^\N \lambda_{\{j\}}^\star \right) \pc{i}{s} + \lambda_{\{i\}}^\star \pc{(i-1)}{s}+\sum_{j=i+1}^\N \lambda_{\{j\}}^\star \rinc{[i+1:\N]}{\{j\}} \\
    &= \pc{i}{s} +\lambda_{\{i\}}^\star \left(\pc{(i-1)}{s}-\pc{i}{s} \right )+\sum_{j=i+1}^\N \lambda_{\{j\}}^\star \left(\rinc{[i+1:\N]}{\{j\}}-\pc{i}{s} \right),
\end{align*}
where in $\rm{(a)}$ we have evaluated $\rinc{[i+1:\N]}{\{j\}}$ using Proposition~\ref{prop:Propfs} and {we have used} the fact that  $\pc{1}{s} \leq \pc{2}{s} \leq \dots \leq \pc{\N}{s}$, and $\rm{(b)}$ follows from the identity $\lambda_\varnothing^\star + \sum_{j=1}^n \lambda_{\{j\}}^\star=1$. 

Using the above equation, we can now recursively express  $\lambda_{\{i\}}^\star$ in terms of $\{\lambda_{\{j\}}^\star: j>i\}$ for $i \in [0:\N]$, given by 
\begin{equation}
\begin{split} \label{lambdarec}
    &\lambda_{\{i\}}^\star \hspace{-1pt}\left (\pc{i}{s}\hspace{-1pt}-\hspace{-1pt}\pc{(i-1)}{s} \right )
     \hspace{-1pt}=\hspace{-1pt} \left (\pc{i}{s}-t^\star \right )+\hspace{-2pt}\sum_{j=i+1}^\N \lambda_{\{j\}}^\star \hspace{-2pt}\left(\rinc{[i+1:\N]}{\{j\}}\hspace{-1pt}-\hspace{-1pt}\pc{i}{s}\hspace{-1pt}\right),
    \end{split}
\end{equation}
where we define ${\pc{0}{s}= \pc{-1}{s}=0}$ and $\lambda_{\{0\}}^\star =0$, for the sake of completeness. 

Before we prove the claim, we present the following lemmas which will be used in the proof. 
\begin{lemma} \label{prop:extra}
If ${\pc{1}{s} \leq \pc{2}{s} \leq \cdots  \leq \pc{\N}{s}}$, then \[{\rinc{[a:n]}{\{j\}} -\pc{(a-1)}{s} \geq \rinc{[b:\N]}{\{j\}}-\pc{(b-1)}{s}},\] 
for all ${1\leq b < a \leq j \leq \N}$.
\end{lemma}

\begin{lemma} \label{singcond}
If ${\pc{1}{s} \leq \pc{2}{s} \leq \cdots \leq \pc{\N}{s}}$ and there exists some ${j \in [\N]}$ such that $\pc{j}{s}=\pc{(j-1)}{s}=\rinc{[j:n]}{\{j\}}$, then $\det{\bP}=0$.
\end{lemma}
The {proof of} Lemma~\ref{prop:extra} and Lemma~\ref{singcond} are presented in~Appendix~\ref{extraproof} and in~Appendix~\ref{singcondproof}, respectively.

We now proceed to show that the $\lambda_{\{i\}}^\star$'s  obtained from~\eqref{lambdarec} are non-negative for all $i \in [\N]$.  The proof is based on contradiction. We start by assuming that the claim is wrong. Let 
 $\mathcal{K}= \left \{i: \lambda_{\{i\}}^\star < 0, i \in [n] \right \}$, and $k$ be the maximum element of $\mathcal{K}$. We also define  $\mathcal{L}= \left \{j\in [1:k-1]:\pc{j}{s}=\pc{(j-1)}{s} \right \}$, and $\ell =\max \mathcal{L} \cup \{0\}$. 
 
Our goal is to show that $\mathcal{K}$ is an empty set and hence, $\lambda_{\{i\}}^\star \geq 0$ for all $i \in [\N]$. 
We first use a {{\em backward induction}} to prove that for every $i\in[\ell+1:k]$ we have $\lambda_{\{i\}}^\star \leq 0$. Note that for the base case of $i=k$, the assumption of $k\in \mathcal{K}$ {implies that} $\lambda_{\{k\}}^\star < 0$.  Now, consider some {$i\in [\ell+1:k-1]$} and assume ${\lambda_{\{i+1\}}^\star, \ldots, \lambda_{\{k\}}^\star \leq 0 }$. Our goal is to show {that} $\lambda_{\{i\}}^\star \leq 0$.
Note that for $\cS=\{j\}$ and $\Omega=[i+1:\N]$ with ${j\in[i+1:\N]}$ we have $\cS^c\cap \Omega^c = [i]$ and hence,  Proposition~\ref{prop:Propfs} implies that ${\rinc{[i+1:\N]}{\{j\}}\geq \max_{x\in [i]}\pc{x}{s} =  \pc{i}{s}}$. Therefore, the coefficients of the form $\left({\rinc{[i+1:\N]}{\{j\}} - \pc{i}{s}}\right)$ in~\eqref{lambdarec} are non-negative. {Then,} starting from~\eqref{lambdarec}, we can write 
\begin{align}\label{eq:lambda-i-vs-k}
    \lambda_{\{i\}}^\star & \left (\pc{i}{s} -\pc{(i-1)}{s} \right )
     \hspace{-1pt}\nonumber\\
     &=\hspace{-1pt} \left (\pc{i}{s}-t^\star \right )+\hspace{-2pt}\sum_{j=i+1}^\N \lambda_{\{j\}}^\star \hspace{-2pt}\left(\rinc{[i+1:\N]}{\{j\}}\hspace{-1pt}-\hspace{-1pt}\pc{i}{s}\hspace{-1pt}\right)\nonumber\\
     &= 
     \left (\pc{i}{s}- t^\star \right )+
     \sum_{j=i+1}^k \lambda_{\{j\}}^\star \hspace{-3pt}\left(\rinc{[i+1:\N]}{\{j\}}\hspace{-1pt}-\hspace{-1pt}\pc{i}{s}\hspace{-1pt}\right)\hspace{-2pt}\nonumber\\
     & \hspace{35mm} +
     \sum_{j=k+1}^\N  \lambda_{\{j\}}^\star \left(\rinc{[i+1:\N]}{\{j\}}-\pc{i}{s}\hspace{-1pt}\right)\nonumber\\
     &\stackrel{{\rm{(a)}}}{\leq}
     \left (\pc{i}{s}-t^\star \right )
     +
     \sum_{j=k+1}^\N \lambda_{\{j\}}^\star \hspace{-2pt}\left(\rinc{[i+1:\N]}{\{j\}}\hspace{-1pt}-\hspace{-1pt}\pc{i}{s}\hspace{-1pt}\right)\nonumber\\
     &\stackrel{{\rm{(b)}}}{\leq}
     \left (\pc{k}{s}-t^\star \right )
     +
     \sum_{j=k+1}^\N \lambda_{\{j\}}^\star \hspace{-2pt}\left(\rinc{[k+1:\N]}{\{j\}}\hspace{-1pt}-\hspace{-1pt}\pc{k}{s}\hspace{-1pt}\right)\nonumber\\
     &= \lambda_{\{k\}}^\star \hspace{-1pt}\left (\pc{k}{s}\hspace{-1pt}-\hspace{-1pt}\pc{(k-1)}{s} \right )\stackrel{{\rm{(c)}}}{\leq} 0,
\end{align}
where $\rm{(a)}$ holds by the induction assumption that $\left \{\lambda_{\{i+1\}}^\star, \ldots, \lambda_{\{k\}}^\star \right \}$ are all non-positive, and $\rm{(b)}$ follows from the fact that  $\pc{i}{s} \leq \pc{k}{s}$ for $i\leq k$, and {applying Lemma~\ref{prop:extra}} for $i<k<j$. 
Finally, since $i>\ell$ we have $\pc{i}{s} -\pc{(i-1)}{s}>0 $, which together with~\eqref{eq:lambda-i-vs-k} implies $\lambda_{\{i\}}^\star \leq 0$. 

Now, let us consider~\eqref{eq:lambda-i-vs-k} for $i=\ell=\max \mathcal{L} \cup \{0\}$. Note that if $\ell\in \mathcal{L}$ we have $\pc{\ell}{s}-\pc{(\ell-1)}{s}=0$, and if $\ell=0$, we get $\lambda_{\{ \ell \}}^\star=0$. This implies that the left-hand-side of~\eqref{eq:lambda-i-vs-k} equals zero. Then, the chain of inequalities in~\eqref{eq:lambda-i-vs-k} is feasible if and only if the three inequalities labeled by $\rm{(a)}$, $\rm{(b)}$, and $\rm{(c)}$ hold with equality.  From $\rm{(b)}$ we can conclude $\pc{\ell}{s} = \pc{k}{s}$, and since $\pc{\cdot}{s}$'s are sorted in an increasing order, we have 
\begin{align}\label{eq:Lem2-1}
    \pc{\ell}{s}=\pc{(\ell+1)}{s}=\ldots = \pc{(k-1)}{s} =\pc{k}{s}. 
\end{align}
Thus, given {the facts that} ${\pc{\ell}{s}=\pc{(\ell+1)}{s}}$ {and} $\ell = \max \mathcal{L}\cup\{0\}$, we can conclude that $\ell=k-1$ (otherwise $\ell+1$ also belongs to $\mathcal{L}$ and hence, $\ell$ cannot be the maximum element of $\mathcal{L} \cup \{0\}$). 

Now, {for $\rm{(a)}$ to hold with equality,} we can conclude that the first summation is zero. However, since {$\lambda_{\{i\}}^\star\leq 0$} for $i\in [\ell+1:k]$ and $\rinc{[i+1:\N]}{\{j\}} -\pc{i}{s} \geq 0$, each term in the summation should be zero. In particular, for the term corresponding to $j=k$, since $\lambda_{\{k\}}^\star<0$, we get 
\begin{align}\label{eq:Lem2-2}
 \pc{\ell}{s} = \rinc{[\ell+1:\N]}{\{k\}} = \rinc{[\ell+1:\N]}{\{\ell+1\}},
\end{align}
where the {second equality} holds since $\ell+1=k$. 
Therefore, from~\eqref{eq:Lem2-1} and \eqref{eq:Lem2-2} we can conclude that the conditions of Lemma~\ref{singcond} are satisfied for $j=\ell+1$, and thus, Lemma~\ref{singcond} implies that
$\text{det}(\bP)=0$. This last conclusion is in contradiction with the assumption of Proposition~\ref{prop:OptUB}.  
{In other words, in order to have $\text{det}(\bP)\neq0$ we need $\mathcal{K}$ to be an empty set and hence, $\lambda_{\{i\}}^\star \geq 0$ for all $i \in [\N]$.}
This completes the proof of {Proposition~\ref{prop:OptUB:1}.}

\section{Proof of Proposition~\ref{prop:mu}}
\label{app:nonnegativelvar}

 We define ${\mathcal{L}= \left \{i \in [\N]:\pc{i}{s}=\pc{(i-1)}{s} \right \}}$ and we let ${\ell= \max \mathcal{L} \cup \{0\}}$. The proof of {Proposition~\ref{prop:mu}} consists of three parts. First, we show that $\mu_j=0$ for all ${j \in [\ell]}$. Then, we prove that all non-zero $\mu_i$'s have the same sign for ${i\in[\ell+1:\N+1]}$. This together with the fact that ${\sum_{i=1}^{\N+1}\mu_i=1}$ guarantees $\mu_i\geq 0$ for all $i \in [\N+1]$. Finally, $\mu_p\geq 0$ is immediately implied from $\mu_p=\sum_{j=1}^{\N+1} \rinc{[j:\N]}{\{i\}} \mu_j$. 

Recall from~\eqref{eq:mu} that  $\mu_i$'s for ${i \in [\N+1]  \cup \{p\}}$ can be obtained by solving the  system of linear equations
\begin{align} \label{eq:mu_transpose}
      \bm{\mu}\pq{}  \!=\! 
\begin{bmatrix} 1  & 0 & \ldots & 0 & 0 \end{bmatrix} \!\!.
\end{align}
Note that for any $k\in[n]$, the $k$th column of the identity in~\eqref{eq:mu_transpose}, is given by 
\begin{align}\label{eq:mu-col-k-1}
    \bm{\mu} \bP_{[0:\N+1],k} =\mu_p - \sum_{j=1}^{\N+1} \rinc{[j:\N]}{\{k\}} \mu_j =0.
\end{align}
Similarly, rewriting the equation for column $n+1$, we get 
\begin{align}\label{eq:mu-col-n+1}
    \bm{\mu} \bP_{[0:\N+1],\N+1}=\mu_p - \sum_{j=1}^{\N+1} \rinc{[j:\N]}{\varnothing} \mu_j =0.
\end{align}
Subtracting~\eqref{eq:mu-col-k-1} from~\eqref{eq:mu-col-n+1}, we get 
\begin{align}\label{eq:Recmu}
0&=   \bm{\mu} \left(\bP_{[0:\N+1],\N+1} - \bP_{[0:\N+1],k}\right) \nonumber\\
&= \sum_{j=1}^{\N+1} \left( \rinc{[j:\N]}{\{k\}} - \rinc{[j:\N]}{\varnothing} \right) \mu_j\nonumber\\
&= \sum_{j=1}^{k} \left (\rinc{[j:\N]}{\{k\}} - \rinc{[j:\N]}{\varnothing} \right ) \mu_j+ \left (\rinc{[k+1:\N]}{\{k\}}-\rinc{[k+1:\N]}{\varnothing} \right)\mu_{k+1} \nonumber\\
        &\quad+ \sum_{j=k+2}^{\N+1} \left (\rinc{[j:\N]}{\{k\}} - \rinc{[j:\N]}{\varnothing} \right)  \mu_j \nonumber\\
        &\stackrel{\rm{(a)}}{=}   \sum_{j=1}^{k} \left (\rinc{[j:\N]}{\{k\}} \!-\! \pc{(j-1)}{s} \right ) \mu_j  \!-\! \left (\pc{k}{s}\!-\!\pc{(k-1)}{s} \right)\mu_{k+1},
        \end{align}
where in $\rm{(a)}$ we used Proposition~\ref{prop:Propfs} to get ${\rinc{[j:\N]}{\varnothing} = \pc{(j-1)}{s}}$, ${\rinc{[k+1:\N]}{\{k\}} = \pc{(k-1)}{s}}$,  and ${\rinc{[j:\N]}{\{k\}} = \pc{(j-1)}{s}}$ for $j\geq k+2$. 

Note that the equation in~\eqref{eq:Recmu} holds for every $k\in[n]$. Moreover, if $k<\ell$, such an equation only involves variables $\{\mu_1,\ldots, \mu_{\ell}\}$. Lastly, for $k=\ell$, we have  $\pc{\ell}{s}\!-\!\pc{(\ell-1)}{s}$, and hence the coefficient of $\mu_{\ell+1}$ will be zero. Thus, equation~\eqref{eq:Recmu} for $k=\ell$ reduces to
\begin{align}\label{eq:Recmu:ell}
\sum_{j=1}^{\ell} \left (\rinc{[j:\N]}{\{\ell\}} - \pc{(j-1)}{s} \right ) \mu_j =0.
\end{align}
This together with the same equation for $k\in[\ell-1]$ {provides} us with a total of $\ell$ equations in $\ell$ variables, namely $\{\mu_1,\ldots, \mu_{\ell}\}$.

Let $\bQ$ be an $(\N+2)\times \ell$ matrix where its $k$th column\footnote{Recall that matrix columns and rows are indexed beginning from $0$.} is given by 
{$\bP_{[0:n+1],n+1} - \bP_{[0:n+1],k+1}$}. 
Note that $\bQ$ is obtained by elementary column {operations} on $\bP$, and since $\bP$ is full-rank, so is $\bQ$, i.e., $\rnk(\bQ)=\ell$. Moreover,~\eqref{eq:Recmu} and~\eqref{eq:Recmu:ell} imply that the $j$th row of $\bQ$ is zero, for $j\in\{0,\ell+1,\ell+2, \ldots, \N+1\}$. Hence, the remaining $\ell$ rows should be linearly independent, which implies {that $\bQ_{[\ell],[0:\ell-1]}$} is full rank. Therefore, the unique solution {for} the system of equations obtained from~\eqref{eq:Recmu} and~\eqref{eq:Recmu:ell}, i.e.,  
\begin{align}
    \begin{bmatrix}
    \mu_1 & \ldots & \mu_\ell
    \end{bmatrix}
    {\bQ_{[\ell],[0:\ell-1]}} = {\mathbf{0}_{1 \times \ell}},
\end{align}
is $\begin{bmatrix}
    \mu_1 & \ldots & \mu_\ell
    \end{bmatrix}={\mathbf{0}_{1 \times \ell}}$.

Next, {we use} induction to show that all non-zero $\mu_{i}$'s have the same sign, for all ${i \in [\ell+1:\N+1]}$. First note that~\eqref{eq:Recmu} for $k=\ell+1$ together with the fact that $\mu_i=0$ for $i\in[\ell]$ implies {that}
\begin{align*}
    0&=\sum_{j=1}^{\ell+1}\left (\rinc{[j:\N]}{\{\ell+1\}} \!-\! \pc{(j-1)}{s} \right ) \mu_j  \!-\! \left (\pc{(\ell+1)}{s}\!-\!\pc{\ell}{s} \right)\mu_{\ell+2}\nonumber\\
    &=\left (\rinc{[\ell+1:\N]}{\{\ell+1\}} \!-\! \pc{\ell}{s} \right ) \mu_{\ell+1}  \!-\! \left (\pc{(\ell+1)}{s}\!-\!\pc{\ell}{s} \right)\mu_{\ell+2}.
\end{align*}
Note that ${\pc{(\ell+1)}{s}\!-\!\pc{\ell}{s}}>0$ since $\ell$ is the maximum element of $\mathcal{L}\cup\{0\}$ and the left side link capacities are arranged in increasing order. Moreover, Proposition~\ref{prop:Propfs} implies that ${\left (\rinc{[\ell+1:\N]}{\{\ell+1\}} \!-\! \pc{\ell}{s} \right ) \geq 0}$. Therefore, we either have $\mu_{\ell+2}=0$, or ${\sgn(\mu_{\ell+2})=\sgn(\mu_{\ell+1})}$. This establishes the base case of the induction. Now, assume {that} our claim holds for every $j\leq k>\ell$, i.e., {all} non-zero $\mu_j$'s have the same sign for $j\leq k$. Then, from~\eqref{eq:Recmu} we have 
\begin{align*}
    \left (\pc{k}{s}-\pc{(k-1)}{s} \right)\mu_{k+1} =\sum_{j=1}^{k} \left (\rinc{[j:\N]}{\{k\}} - \pc{(j-1)}{s} \right ) \mu_j. 
\end{align*}
Similar to the base case, we note that ${\pc{k}{s}-\pc{(k-1)}{s} >0}$, and ${\rinc{[j:\N]}{\{k\}} - \pc{(j-1)}{s}\geq 0}$. 
{Therefore, $\mu_{k+1}$ is either zero or its sign is identical to the one of the non-zero $\mu_j$'s with $j \in [\ell+1:k]$}
This completes the induction, from which we can conclude that 
{all non-zero $\mu_i$'s have the same sign for ${i\in[\ell+1:\N+1]}$.}
Finally, since $\sum_{k=1}^{\N+1}\mu_k=1$, this common sign has to be positive. 

Lastly, from $\mu_p=\sum_{j=1}^{\N+1} \rinc{[j:\N]}{\{i\}} \mu_j$ for all $i \in [\N+1]$, it directly follows that $\mu_p \geq 0$.
This concludes the proof of {Proposition~\ref{prop:mu}.}

\section{Proof of Lemma~\ref{prop:extra}} \label{extraproof}
Recall that $\inc{[a:\N]}{\{j\}}$ for  $j\in [a:\N]$ is the transfer matrix from $X_{\{s,j\}}$ to $Y_{\{d,1,2,\dots, a-1\}}$, and {it is given by}
\begin{align*}
\inc{[a:n]}{\{j\}}
=
  \left[ 
\begin{array}{c|c} 
  {\mathbf{0}_{\eta \times \eta}} & \mathbf{D}^{\eta-\pc{d}{j}} \\
  \hline 
  \mathbf{D}^{\eta-\pc{1}{s}} & \mathbf{D}^{\eta-\pc{1}{j}} \\
  \hline 
  \mathbf{D}^{\eta-\pc{2}{s}} & \mathbf{D}^{\eta-\pc{2}{j}} \\
  \hline 
  \vdots& \vdots \\
  \hline 
  \mathbf{D}^{\eta-\pc{(a-1)}{s}} & \mathbf{D}^{\eta-\pc{(a-1)}{j}}
\end{array} 
\right].
\end{align*}
Similarly,  {for  $1 \leq b < a$,} the matrix $\inc{[b:\N]}{\{j\}}$ is given by the top $b$ block rows of $\inc{[a:\N]}{\{j\}}$. 
Since ${\pc{1}{s} \leq \pc{2}{s} \leq \cdots \leq \pc{\N}{s}}$, from the definition of $\mathbf{D}^{\eta-m}$ in~\eqref{eq:DiamLDA} and recalling that we index rows and columns of a matrix starting from zero, it follows that columns ${[\pc{(b-1)}{s}: \pc{}{}-1]}$  in $\inc{[b:\N]}{\{j\}}$ are zero, {where $\pc{0}{s}=0.$}

Now, consider the {lowest} left block of $\inc{[a:\N]}{\{j\}}$, namely $\mathbf{D}^{\eta-\pc{(a-1)}{s}}$. 
{
From~\eqref{eq:DiamLDA}, for every $\ell \in { [\pc{(b-1)}{s}:\pc{(a-1)}{s}-1]}$, the row ${\pc{}{}-\pc{(a-1)}{s}+\ell}$ of the matrix $\mathbf{D}^{\eta-\pc{(a-1)}{s}}$ has a one in column $\ell$ and zero elsewhere. 
Since $\inc{[b:\N]}{\{j\}}$ is fully zero in these columns, the row ${\pc{}{}-\pc{(a-1)}{s}+\ell}$ of the matrix $\inc{[a:\N]}{\{j\}}$ is linearly independent from all rows in $\inc{[b:\N]}{\{j\}}$. 
Thus,} $\inc{[a:\N]}{\{j\}}$ has at least ${\pc{(a-1)}{s}-\pc{(b-1)}{s}}$ additional  linearly independent rows compared to $\inc{[b:\N]}{\{j\}}$, which immediately implies {that}
${\rinc{[a:n]}{\{j\}} \geq \rinc{[b:\N]}{\{j\}}+\left (\pc{(a-1)}{s}-\pc{(b-1)}{s} \right )}$. This concludes the proof of Lemma~\ref{prop:extra}.

\section{Proof of Lemma~\ref{singcond}} \label{singcondproof}
The outline of the proof is as follows: we show that, if ${\pc{1}{s} \leq \pc{2}{s} \leq \cdots \leq \pc{\N}{s}}$ and {$\pc{j}{s}=\pc{(j-1)}{s}=\rinc{[j:n]}{\{j\}}$} for some $j\in[n]$,  then columns $j$ and $\N+1$ of the matrix $\bP$ are identical, i.e., $\bP$ is singular and hence, $\det{\bP}=0$.
Towards this end, we start by noting that
${\PQ{i}{j} = -\rinc{[i:\N]}{\{j\}}}$ for ${(i,j)\in[n+1]\times [\N+1]}$  and the column $j$ of the matrix $\bP$ is given by 
\begin{align*}
    \bP_{[0:\N+1],j}\!=\!\begin{bmatrix}
    {1} &\!\! -\rinc{[1:\N]}{\{j\}} &\! -\rinc{[2:\N]}{\{j\}} &\!\! \cdots &\!\! -\rinc{[\N:\N]}{\{j\}} &\! -\rinc{[\N+1:\N]}{\{j\}}
    \end{bmatrix}^T\!\!. 
\end{align*}
Now, we evaluate each entry of the vector $\bP_{[0:\N+1],j}$. Consider some $i \in [j+2:\N+1]$. Using Proposition~\ref{prop:Propfs} for $\cS=\{j\}$ and $\Omega=[i:\N]$ with $\Omega \cap \cS= \varnothing$ we have
\begin{align}\label{eq:lm2:1}
{\rinc{[i:\N]}{\{j\}}=\max_{t\in [1:i-1]\setminus\{j\}} \pc{t}{s} = \pc{(i-1)}{s}}.
\end{align} 
Similarly,  for $i=j+1$  using Proposition~\ref{prop:Propfs} we get
\begin{align}\label{eq:lm2:2}
{\rinc{[i:\N]}{\{j\}}= \max_{t\in [1:i-1]\setminus\{j\}} \pc{t}{s} =\pc{(j-1)}{s} \stackrel{{\rm{(a)}}}{=}  \pc{j}{s} =\pc{(i-1)}{s}},
\end{align}
where the equality in \rm{(a)} follows from the assumption of the lemma. 
It remains to evaluate $\rinc{[i:n]}{\{j\}}$ for $i\in [j]$. 
Consider $\inc{[j:n]}{\{j\}}$ which is defined as
\begin{align*}
\inc{[j:n]}{\{j\}}
=
  \left[ 
\begin{array}{c|c} 
  {\mathbf{0}_{\eta \times \eta}} & \mathbf{D}^{\eta-\pc{d}{j}} \\
  \hline 
  \mathbf{D}^{\eta-\pc{1}{s}} & \mathbf{D}^{\eta-\pc{1}{j}} \\
  \hline 
  \mathbf{D}^{\eta-\pc{2}{s}} & \mathbf{D}^{\eta-\pc{2}{j}} \\
  \hline 
  \vdots& \vdots \\
  \hline 
  \mathbf{D}^{\eta-\pc{(j-2)}{s}} & \mathbf{D}^{\eta-\pc{(j-2)}{j}}\\
  \hline 
  \mathbf{D}^{\eta-\pc{(j-1)}{s}} & \mathbf{D}^{\eta-\pc{(j-1)}{j}}
\end{array} 
\right],
\end{align*}
where $\mathbf{D}^{\eta-m}$ is given in~\eqref{eq:DiamLDA}. Focusing on the first column-block of $\inc{[j:n]}{\{j\}}$,  i.e., 
\[
\left[ 
\begin{array}{c|c|c|c} 
  \!\!\mathbf{0}_{\eta \times \eta} \!&\! \mathbf{D}^{\eta-\pc{1}{s}} \!&\!
  \dots \!&\! \mathbf{D}^{\eta-\pc{(j-1)}{s}}\!\!\!\!
\end{array}
\right]^T,
\]
we observe that each row in this $\eta \times j\eta$ matrix is either zero or appearing in its lowest block, $\mathbf{D}^{\eta-\pc{(j-1)}{s}}$. Hence, we have 
\begin{align*}
    \rnk&\!\left(\!\left[ 
\begin{array}{c|c|c|c} 
  \!\!\mathbf{0}_{\eta \times \eta} \!\!&\!\! \mathbf{D}^{\eta-\pc{1}{s}} \!\!&\!\!
  \dots \!\!&\!\! \mathbf{D}^{\eta-\pc{(j-1)}{s}}\!\!\!\!
\end{array}
\right]^T\right) \!=\! \rnk \!\left(\mathbf{D}^{\eta-\pc{(j-1)}{s}}\right)\\ 
&= \pc{(j-1)}{s} \stackrel{{\rm{(a)}}}{=} \rinc{[j:n]}{\{j\}} \stackrel{{\rm{(b)}}}{=}\rnk \left(\inc{[j:n]}{\{j\}}\right),
\end{align*}
where in~$\rm{(a)}$ we used  the assumption of the lemma, and~$\rm{(b)}$ follows from~\eqref{eq:def:pc}. In other words, the (column)-rank of $\inc{[j:n]}{\{j\}}$ equals the (column)-rank of its first block column, or equivalently, every 
 column in the second column block of $\inc{[j:n]}{\{j\}}$  can be written as a linear combination of the columns in the first column block of $\inc{[j:n]}{\{j\}}$. Next, note that for every $i\in[j]$ the matrix 
\begin{align*}
\inc{[i:n]}{\{j\}}
&=
  \left[ 
\begin{array}{c|c} 
  {\mathbf{0}_{\eta \times \eta}} & \mathbf{D}^{\eta-\pc{d}{j}} \\
  \hline 
  \mathbf{D}^{\eta-\pc{1}{s}} & \mathbf{D}^{\eta-\pc{1}{j}} \\
  \hline 
  \mathbf{D}^{\eta-\pc{2}{s}} & \mathbf{D}^{\eta-\pc{2}{j}} \\
  \hline 
  \vdots& \vdots \\
  \hline 
  \mathbf{D}^{\eta-\pc{(i-2)}{s}} & \mathbf{D}^{\eta-\pc{(i-2)}{j}}\\
  \hline 
  \mathbf{D}^{\eta-\pc{(i-1)}{s}} & \mathbf{D}^{\eta-\pc{(i-1)}{j}}
\end{array} 
\right]
\end{align*}
is a sub-matrix of $\inc{[j:n]}{\{j\}}$ and hence, the same conclusion holds for $\inc{[i:n]}{\{j\}}$, that is,  each column in the second column block of $\inc{[i:n]}{\{j\}}$ is also a linear combination of the columns in the first column block of $\inc{[i:n]}{\{j\}}$. Therefore, the rank of $\inc{[i:n]}{\{j\}}$ equals the rank of its first column block, which leads to 
\begin{align}\label{eq:lm2:3}
\rinc{[i:n]}{\{j\}}  &= \text{rank}\left(
\inc{[i:n]}{\{j\}}\right) \nonumber\\
&= 
\rnk\left(\left[ 
\begin{array}{c|c|c|c} 
  \!\!\mathbf{0}_{\eta \times \eta} \!&\! \mathbf{D}^{\eta-\pc{1}{s}} \!&\!
  \dots \!&\! \mathbf{D}^{\eta-\pc{(i-1)}{s}}\!\!\!\!
\end{array}
\right]^T\right) \nonumber\\
&=\max_{y\in [i-1]} \text{rank}\left(\mathbf{D}^{\eta-\pc{y}{s}}\right) \nonumber\\
& \stackrel{{\rm{(a)}}}{=} \max_{y\in[i-1]} \pc{y}{s} \stackrel{{\rm{(b)}}}{=} \pc{(i-1)}{s},
\end{align}
where 
$\rm{(a)}$ is due to the fact that  the rank of  $\mathbf{D}^{\eta-m}$ equals $m$, and $\rm{(b)}$ follows since ${\pc{1}{s} \leq \pc{2}{s} \leq \cdots \leq \pc{\N}{s}}$.
Therefore, using~\eqref{eq:lm2:1}--\eqref{eq:lm2:3} the entries of $\bP
_{[0:\N+1],j}$ can be evaluated as
\begin{align*}
    \bP
_{[0:\N+1],j}\!=\!\begin{bmatrix}
    {1} \!\!&\!\! {-\pc{0}{s}}  \!&\! -\pc{1}{s} \!&\! -\pc{2}{s} \!&\!\! \cdots \!\!&\! -\pc{(\N-1)}{s} & -\pc{\N}{s}
    \end{bmatrix}^T\!\!,
\end{align*}
{where  $\pc{0}{s} = 0$.}

Finally,  using~\eqref{eq:MatrP} for $j=\N+1$ and Proposition~\ref{prop:Propfs} for $\cS=[\N+1:\N]=\varnothing$ and $\Omega=[i:\N]$ with ${\Omega \cap \cS =\varnothing}$ we get 
\begin{align*}
    \PQ{i}{n+1} = -\rinc{[i:n]}{\{n+1\}} = -\max_{t\in [1:i+1]} \pc{(i-1)}{s} \pc{t}{s} = 
    - \pc{(i-1)}{s}. 
\end{align*}
Therefore, the $(\N+1)$th column of $\PQ{}{}$ is identical to its $j$th column. This concludes the proof of Lemma~\ref{singcond}.


\begingroup
\let\cleardoublepage\clearpage
 \bibliography{BibSJ1.bib}

\begin{thebibliography}{10}
\providecommand{\url}[1]{#1}
\csname url@samestyle\endcsname
\providecommand{\newblock}{\relax}
\providecommand{\bibinfo}[2]{#2}
\providecommand{\BIBentrySTDinterwordspacing}{\spaceskip=0pt\relax}
\providecommand{\BIBentryALTinterwordstretchfactor}{4}
\providecommand{\BIBentryALTinterwordspacing}{\spaceskip=\fontdimen2\font plus
\BIBentryALTinterwordstretchfactor\fontdimen3\font minus
  \fontdimen4\font\relax}
\providecommand{\BIBforeignlanguage}[2]{{%
\expandafter\ifx\csname l@#1\endcsname\relax
\typeout{** WARNING: IEEEtran.bst: No hyphenation pattern has been}%
\typeout{** loaded for the language `#1'. Using the pattern for}%
\typeout{** the default language instead.}%
\else
\language=\csname l@#1\endcsname
\fi
#2}}
\providecommand{\BIBdecl}{\relax}
\BIBdecl

\bibitem{CardoneTIT2016}
M.~Cardone, D.~Tuninetti, and R.~Knopp, ``On the optimality of simple schedules
  for networks with multiple half-duplex relays,'' \emph{IEEE Trans. Inf.
  Theory}, vol.~62, no.~7, pp. 4120--4134, July 2016.

\bibitem{AvestimehrIT2011}
A.~S. Avestimehr, S.~N. Diggavi, and D.~N.~C. Tse, ``Wireless network
  information flow: A deterministic approach,'' \emph{IEEE Trans. Inf. Theory},
  vol.~57, no.~4, pp. 1872--1905, April 2011.

\bibitem{OzgurIT2013}
A.~{\"O}zg{\"u}r and S.~N. Diggavi, ``Approximately achieving {G}aussian relay
  network capacity with lattice-based {QMF} codes,'' \emph{IEEE Trans. Inf.
  Theory}, vol.~59, no.~12, pp. 8275--8294, December 2013.

\bibitem{LimIT2011}
S.~Lim, Y.-H. Kim, A.~El~Gamal, and S.-Y. Chung, ``Noisy network coding,''
  \emph{IEEE Trans. Inf. Theory}, vol.~57, no.~5, pp. 3132 --3152, 2011.

\bibitem{LimISIT2014}
S.~H. Lim, K.~T. Kim, and Y.~H. Kim, ``Distributed decode-forward for
  multicast,'' in \emph{IEEE International Symposium on Information Theory
  (ISIT)}, June 2014, pp. 636--640.

\bibitem{CardoneIT2014}
M.~Cardone, D.~Tuninetti, R.~Knopp, and U.~Salim, ``Gaussian half-duplex relay
  networks: improved constant gap and connections with the assignment
  problem,'' \emph{IEEE Trans. Inf. Theory}, vol.~60, no.~6, pp. 3559 -- 3575,
  June 2014.

\bibitem{bagheri2014}
H.~{Bagheri}, A.~S. {Motahari}, and A.~K. {Khandani}, ``On the capacity of the
  half-duplex diamond channel,'' in \emph{2010 IEEE International Symposium on
  Information Theory}, 2010, pp. 649--653.

\bibitem{JainITW2021}
S.~Jain, M.~Cardone, and S.~Mohajer, ``Operating half-duplex diamond networks
  with two interfering relays,'' in \emph{IEEE Information Theory Workshop
  (ITW)}, July 2021.

\bibitem{EzzeldinISIT2017}
Y.~H. Ezzeldin, M.~Cardone, C.~Fragouli, and D.~Tuninetti, ``Efficiently
  finding simple schedules in {G}aussian half-duplex relay line networks,'' in
  \emph{IEEE International Symposium on Information Theory (ISIT)}, June 2017,
  pp. 471--475.

\bibitem{EtkindTIT2014}
R.~H. Etkin, F.~Parvaresh, I.~Shomorony, and A.~S. Avestimehr, ``Computing
  half-duplex schedules in {G}aussian relay networks via min-cut
  approximations,'' \emph{IEEE Trans. Inf. Theory}, vol.~60, no.~11, pp.
  7204--7220, November 2014.

\bibitem{JainISIT2019}
S.~Jain, M.~Elyasi, M.~Cardone, and S.~Mohajer, ``On simple scheduling in
  half-duplex relay diamond networks,'' in \emph{IEEE International Symposium
  on Information Theory (ISIT)}, July 2019.

\end{thebibliography}
 \bibliographystyle{IEEEtran}
\endgroup

\end{document}